\documentclass[letterpaper]{article}
\usepackage{aaai2026} % DO NOT CHANGE THIS
\usepackage{times} % DO NOT CHANGE THIS
\usepackage{helvet} % DO NOT CHANGE THIS
\usepackage{courier} % DO NOT CHANGE THIS
\usepackage[hyphens]{url} % DO NOT CHANGE THIS
\usepackage{graphicx} % DO NOT CHANGE THIS
\usepackage{graphicx} % DO NOT CHANGE THIS
\usepackage{natbib} % DO NOT CHANGE THIS
\usepackage{caption} % DO NOT CHANGE THIS
\usepackage[utf8]{inputenc} % allow utf-8 input
\usepackage[T1]{fontenc}    % use 8-bit T1 fonts
\usepackage{booktabs}       % professional-quality tables
\usepackage{amsfonts}       % blackboard math symbols
\usepackage{nicefrac}       % compact symbols for 1/2, etc.
\usepackage{microtype}      % microtypography
\usepackage{graphicx} % Required for inserting images
\usepackage{amsmath,amsfonts,amsthm,mathtools,bm,subcaption}
\usepackage[ruled,linesnumbered,vlined,noend]{algorithm2e}
\usepackage{appendix}
\usepackage{thmtools}
\usepackage[capitalise,noabbrev]{cleveref}

\usepackage[super]{nth}

\newcommand{\tr}{^{\mathsf{T}}}
\newcommand{\Real}{\mathbb{R}}
\newcommand{\Nats}{\mathbb{N}}
\newcommand{\E}{\mathbb{E}}
\newcommand{\opt}{^{\star}}
\newcommand{\probP}{\text{I\kern-0.15em P}}
\DeclareMathOperator*{\argmin}{argmin}

\newcommand{\saddle}{\mathfrak{S}}
\newcommand{\rhog}{\rho_{\mathrm{G}}}
\newcommand{\rhor}{\rho_{\mathrm{R}}}

\renewcommand{\cite}{\citep}

\theoremstyle{plain}
\newtheorem{theorem}{Theorem}[section]
\newtheorem{proposition}[theorem]{Proposition}
\newtheorem{lemma}[theorem]{Lemma}

\theoremstyle{definition}

\theoremstyle{remark}
\newtheorem{remark}[theorem]{Remark}
\newtheorem{example}[theorem]{Example}

\title{Convergence of Fast Policy Iteration in Markov Games and Robust MDPs}

\author{
  Keith Badger\textsuperscript{\rm 1}, 
  Jefferson Huang\textsuperscript{\rm 2},
  Marek Petrik\textsuperscript{\rm 3}
}
\affiliations{
  \textsuperscript{\rm 1}University of New Hampshire, keith.badger@unh.edu, \\
  \textsuperscript{\rm 2}Naval Postgraduate School, jefferson.huang@nps.edu, \\
  \textsuperscript{\rm 3}University of New Hampshire, mpetrik@cs.unh.edu
}

\begin{document}

\nocopyright

\maketitle

\begin{abstract}
Markov games and robust MDPs are closely related models that involve computing a pair of saddle point policies. As part of the long-standing effort to develop efficient algorithms for these models, the Filar-Tolwinski (FT) algorithm has shown considerable promise. As our first contribution, we demonstrate that FT may fail to converge to a saddle point and may loop indefinitely, even in small games. This observation contradicts the proof of FT's convergence to a saddle point in the original paper. As our second contribution, we propose Residual Conditioned Policy Iteration (RCPI). RCPI builds on FT, but is guaranteed to converge to a saddle point. Our numerical results show that RCPI outperforms other convergent algorithms by several orders of magnitude.
\end{abstract}

\section{Introduction}

Markov Games~(MG)~\cite{Kallenberg2022} and Robust MDPs~(RMDPs)~\cite{Iyengar2005, Wiesemann2013, Ho2022} are two important models that generalize Markov Decision Processes~(MDPs)~\cite{Puterman2005}. Markov games can model strategic adversaries that can act to minimize the agent's returns and are a common model in multi-agent reinforcement learning~\cite{Shou2022,Littman1994a}. Similarly, RMDPs can model an adversarial nature that can perturb transition probabilities and rewards to minimize the agent's returns and are useful when making decisions with imperfect data-driven models~\cite{Lobo2023, Behzadian2021}. In recent years, MGs and RMDPs have seen an increasing number of applications in machine and reinforcement learning, which has motivated the study of efficient algorithms for solving them~\cite{Perolat2016, Ho2021, Ho2022, Behzadian2021a, Kaufman2013, Winnicki2023}.

Although basic algorithms, like value and policy iteration, adapt readily from MDPs to MGs and RMDPs, developing more efficient algorithms has been challenging. The efforts to adapt efficient optimistic policy iteration~(OPI) algorithms, such as modified or fitted policy iteration, have been difficult. Many natural OPI algorithms proposed for MGs and RMDPs cycle among suboptimal policies, alternatively improving the minimization or maximization sides of the saddle point equilibrium. The lack of convergence is often counter-intuitive and has led to several incorrect convergence proofs in the literature~\cite{Condon1993,Perolat2016,Filar1991}. The overarching reason is that OPI in MG and RMDPs do not monotonically improve the policy and its value function as in MDPs.

We make two main contributions in this paper. First, we show that a fast OPI method proposed in~\citet{Filar1991} can terminate with an arbitrarily suboptimal policy. \citet{Perolat2016} first identified a gap in the proof of correctness in \citet{Filar1991} but hypothesized the algorithm works nevertheless. In contrast, we show that the algorithm is inherently suboptimal. 

Second, we propose and analyze \emph{Residual Conditioned Policy Iteration}~(RCPI). RCPI is a new, simple approximate policy iteration algorithm for solving MGs and RMDPs that is guaranteed to converge to optimal policies. It builds on earlier efficient OPI algorithms~\cite{Filar1991,Perolat2016,Ho2021,Winnicki2023} and combines them with an adaptive correction step. Our theoretical analysis shows that RCPI matches the worst-case computational complexity of value iteration. Our numerical results show that on a wide range of problems, RCPI outperforms other convergent algorithms by several orders of magnitude, even in moderately sized problems. 

In this paper, we restrict our focus to model-based algorithms for MGs and RMDPs. It is important to note that this setting differs from online algorithms for solving games and multi-agent reinforcement learning problems, such as in~\citet{Zhang2022}. Although some of the issues that need to be overcome in online and model-based solvers are similar, we leave the study of the exact relationship between online and model-based algorithms for future work. 

The remainder of the paper is organized as follows. \Cref{sec:prior-work} positions our work in the context of prior algorithmic developments for MGs and RMDPs. Then, \cref{sec:background} describes the formal framework for MGs and RMDPs. \Cref{sec:filar-tolw-suboptimal} describes our first contribution, which is to show that an existing OPI algorithm~\cite{Filar1991} may fail with an arbitrarily suboptimal policy. \Cref{sec:rcpi:-resid-corr} describes our second and main contribution, the RCPI algorithm, along with its convergence rate and computational complexity analysis. Finally, our numerical results in \cref{sec:numerical-results} compare RCPI with existings algorithms for solving MGs and RMDPs. 

% \mm{there is some inconsitency in the terms we are using. The agents are primary, secondary, adversarial, maximizing, minimizing. It may be better to pick some terms and stick to them throughout.}

\section{Prior Work: Solving MGs and RMDPs} \label{sec:prior-work}

In this section, we summarize prior efforts on developing OPI algorithms for MGs and RMDPs. We note that MG and RMDP communities have been largely separate, though the similarities between them have been noted and exploited previously~\cite{Iyengar2005,Grand-Clement2024,Grand-Clement2025}.
 
Value iteration is a simple convergent algorithm for solving MDPs, RMDPs, and MGs, but can be very slow in many practical settings~\cite{Puterman2005}. Many faster convergent algorithms for MDPs exist, such as policy iteration or modified policy iteration. Since the early days of MG~\cite{Condon1993} and RMDPs~\cite{Iyengar2005,Kaufman2013}, researchers have sought to generalize the ideas of modified policy iteration from MDPs to MGs and RMDPs. However, the attempts to speed up policy iteration while guaranteeing convergence have been largely unsuccessful~\cite{Perolat2016}. Existing algorithms are either too slow for larger problems or lack optimality guarantees.

Policy iteration~\cite{Puterman2005}, another basic MDP algorithm, can dramatically reduce the number of Bellman operator evaluations and compute the optimal policy in strongly polynomial time in MDPs~\cite{Ye2011}. Hoffman-Karp algorithm, also known as robust policy iteration~\cite{Iyengar2005}, for MGs and RMDPs adapts policy iteration to MGs and RMDPs. Although Hoffman-Karp has polynomial worst-case time complexity~\cite{Hansen2013}, it can be slower than value iteration in practice. Each Hoffman-Karp policy evaluation requires computing the adversarial agent's optimal policy. That is a significant increase in the complexity of the policy evaluation step in MDPs, which entails solving a system of linear equations.

Optimistic policy iteration~(OPI) methods, such as modified policy iteration, accelerate policy iteration by performing the evaluation step approximately~\cite{Puterman2005}. In MDPs, OPI algorithms dramatically improve empirical performance while preserving the worst-case convergence rate of value iteration. In MGs and RMDPs, many natural OPI algorithms attain good empirical performance but fail to compute optimal policies~\cite{Condon1993}. For instance, Pollatschek Avi-Itzhak~(PAI) algorithm holds the adversarial policy constant in the policy evaluation step, which is quicker than Hoffman-Karp, but may lead to infinitely looping over suboptimal policies~\cite{VanderWal1978}. 

One well-known attempt to fix PAI's non-convergence is the Filar-Tolwinski~(FT) algorithm~\cite{Filar1991}. It leverages the observation that PAI can be seen as Newton's method on the $L_2$ norm of the Bellman residual. FT replaces the pure Newton's method of PAI with the modified Newton's method, which uses Armijo's rule when deciding the step size in the value function update. While~\cite{Filar1991} claims that this resolves the cycling issues found with PAI, we show that FT may not converge. We discuss this issue in more detail in \cref{sec:filar-tolw-suboptimal}.

Recent years have seen several notable attempts to develop algorithms that match the empirical performance of PAI while guaranteeing convergence to an optimal policy. Robust Modified Policy Iteration~(RMPI)~\cite{Kaufman2013} and Partial Policy Iteration~(PPI)~\cite{Ho2021} modify Hoffman-Karp to evaluate the adversarial policy approximately. RMPI uses a fixed-precision approximation, while PPI adapts the evaluation throughout the algorithm's execution. Numerical evidence suggests that PPI outperforms RMPI~\cite{Ho2021}. The Winnicki-Srikant~(WS) algorithm combines value iteration steps with policy backup steps and proposes ratios that guarantee the algorithm's convergence~\cite{Winnicki2023}. 

\section{Preliminaries: MGs and Robust MDPs}
\label{sec:background}

In this section, we define Markov games and robust Markov Decision Processes formally and describe the properties we use to derive our main results.

\subsection{Notation} The symbols $\Real$ and $\Nats$ denote the sets of real and natural (including $0$) numbers.
Vectors are denoted with a lower-case bold font, such as $\bm{x} \in \Real^n$, and $x_i, i = 1, \dots , n$ is the $i$-th element of the vector. Matrices are denoted in uppercase bold font. Sets are denoted with calligraphic letters. We use the notation $\Real^{\mathcal{Z}}$ to denote the set of all functions $f\colon \Real \mapsto \mathcal{Z}$, and interpret each $f$ equivalently as a vector $\bm{f}$ such that $f_z = f(z)$. The notation $\Delta^\mathcal{Z} := \left\{ \bm{x} \in\Real^\mathcal{Z} \mid  \bm{1}\tr \bm{x} = 1, \bm{x} \geq \bm{0} \right\}$ refers to the set of probability distributions over the finite non-empty set $\mathcal{Z}$.

To streamline our notation, we define the \emph{$\epsilon$-saddle-point operator} $\saddle_{\epsilon}\colon \Real^{\mathcal{X} \times \mathcal{Y}} \to 2^{\mathcal{X} \times \mathcal{Y}}$ for any tolerance $\epsilon \ge 0$ and an objective function $f\colon \mathcal{X} \times  \mathcal{Y} \to \Real $ as
\begin{equation} \label{eq:saddle-point}
  \begin{gathered}
  \saddle_{\epsilon}(f)
  \;:=\; 
  \Bigl\{  (x\opt,y\opt) \in \mathcal{X} \times \mathcal{Y} \mid  \\
  f(x,y\opt) - \epsilon  \leq f(x\opt,y\opt) \leq f(x\opt,y) + \epsilon, \\
  \; \forall x \in \mathcal{X}, y\in \mathcal{Y} \Bigr\},
  \end{gathered}
\end{equation}
where $\mathcal{X}, \mathcal{Y}$ are arbitrary sets. Note that the first parameter of $f$ is maximized, and the second one is minimized. The intuitive explanation of this definition is that $x\opt$ and $y\opt$ are $\epsilon$-optimal responses to each other. In the remainder of the paper, we shorten $\saddle := \saddle_0$. Note that if $\epsilon_1 \le \epsilon_2$ then $\saddle_{\epsilon_1}(f) \subseteq \saddle_{\epsilon_2}(f)$. We also allow $\saddle_0$ to be used with objective function $f\colon \mathcal{X} \times \mathcal{Y} \to \mathcal{Z}$ for some partially ordered set $\mathcal{Z}$. 

If the function $f$ is real-valued and bi-linear, then an element of $\saddle(f)$ can be computed using the standard linear problem formulation of matrix games, see for example~\cite[section~10.1.3]{Kallenberg2022}.

\subsection{Markov Games}
Markov games extend Markov decision processes to a zero-sum game-theoretic setting~\cite{Kallenberg2022,Filar1996} and can model multi-agent reinforcement learning. An imperfect-information \emph{Markov game} is defined as $(\mathcal{S}, \mathcal{A}, \mathcal{B}, r, P, s_0)$ where $\mathcal{S} = \left\{ 1, \dots , S \right\}$ is the finite non-empty set of states that the agents share, $\mathcal{A} = \left\{ 1, \dots , A \right\}$ is the finite non-empty set of actions for the primary agent, $\mathcal{B} = \left\{ 1, \dots , B \right\}$ is the finite non-empty set of actions for the adversarial agent. The function $r\colon  \mathcal{S} \times \mathcal{A} \times \mathcal{B}  \to [-r_{\max},r_{\max}]$ for $r_{\max} \in \Real$ represents the rewards the primary agent seeks to maximize and the adversarial agent seeks to minimize. The function $P\colon  \mathcal{S} \times \mathcal{A} \times \mathcal{B} \to \Delta^{\mathcal{S}}$ is the transition probability function, where $p(s, a, b, s')$ is the probability of transitioning from state $s$ to state $s'$ after the agents take actions action $a$ and $b$, respectively. Finally, $s_0 \in \mathcal{S}$ is the initial state.

We consider \emph{infinite-horizon discounted rewards} for a discount factor $\gamma \in (0,1)$, and restrict attention to randomized stationary policies $\Pi := (\Delta^{\mathcal{A}})^{\mathcal{S}}$ and $\Sigma := (\Delta^{\mathcal{B}})^{\mathcal{S}}$ for the maximizing and minimizing agents, respectively. Note that the restriction to randomized stationary policies is not limiting, because neither of the players can benefit from using Markov or history-dependent policies~\cite{Kallenberg2022, Filar1996}.
The value function $\bm{v}^{\bm{\pi}, \bm{\sigma}} \in \Real^S$ associated with each $\bm{\pi}\in \Pi$ and $\bm{\sigma}\in \Sigma$ as~\cite{Filar1996}:
\begin{equation} \label{eq:value-policy}
v^{\bm{\pi},\bm{\sigma}}_s
  \; :=\; 
  \E_{\bm{\pi}, \bm{\sigma}}^s
  \left[ \sum_{t=0}^{\infty} \gamma^t r(\tilde{s}_t, \tilde{a}_t, \tilde{b}_t)\right],
  \quad
  \forall s\in \mathcal{S}.
\end{equation}
The superscripts and subscripts of $\E^s_{\bm{\pi}, \bm{\sigma}}$ indicate that the probability measure is chosen such that $\tilde{s}_0 = s$ and that $\tilde{a}_t \sim \bm{\pi}(\tilde{s}_t)$, $\tilde{b}_t \sim \bm{\sigma}(\tilde{s}_t)$, and $\tilde{s}_{t+1} \sim P(\tilde{s}_t, \tilde{a}_t, \tilde{b}_t, \cdot)$ for all $t\in \Nats$. In general, we adorn random variables with a tilde. The equilibrium value function $\bm{v\opt} \in \Real^S$ is defined as the saddle point over policy pairs:
\begin{equation} \label{eq:value-optimal}
  v\opt_s
  \; :=\;
  \max_{\bm{\pi}\in \Pi } \min_{\bm{\sigma} \in \Sigma } \, v^{\bm{\pi}, \bm{\sigma} }_s, \qquad \forall s\in \mathcal{S}.
\end{equation}

 That is, the agents seek to compute the saddle point of the infinite-horizon discounted objective function $\rhog \colon \mathcal{S} \times \Pi \times  \Sigma \to \Real$ for some tolerance $\epsilon \ge 0$: 
\begin{equation} \label{eq:objective-function}
  (\bm{\pi\opt}, \bm{\sigma\opt}) \in \saddle_{\epsilon}(\rho),
  \; \text{where} \;
  \rhog(s_0, \bm{\pi}, \bm{\sigma})
  \; :=\; 
  v_{s_0}^{\bm{\pi}, \bm{\sigma}}.
  %\lim_{T\to\infty}
  % \E_{\bm{\pi}, \bm{\sigma}}^{s_0} \left[ \sum_{t=0}^{\infty} \gamma^t r(\tilde{s}_t, \tilde{a}_t, \tilde{b}_t)\right].
\end{equation}
Given any $\epsilon \geq 0$, the existence of an equilibrium pair $(\bm{\pi\opt} , \bm{\sigma\opt})$ is guaranteed for discounted Markov games with finite state and action sets~\cite[corollary 10.1]{Kallenberg2022}.

Next, we describe the \emph{Bellman operator} for Markov games. For each $\bm{\pi}\in \Pi$ and $\bm{\sigma}\in \Sigma$, we define the reward vector $\bm{r}^{\bm{\pi},\bm{\sigma}}\in \Real^S$ and a transition matrix $\bm{P}^{\bm{\pi},\bm{\sigma}} \in \Real_{+}^{S \times  S}$ as
 \begin{align*}
   \bm{r}^{\bm{\pi},\bm{\sigma}}_s
   &:= \smashoperator[r]{\sum_{(a,b) \in \mathcal{A} \times  \mathcal{B}}} \pi_a(s) \cdot \sigma_b(s) \cdot r(s,a,b), \\
   \bm{P}^{\bm{\pi},\bm{\sigma}}_{s,s'}
   &:= \smashoperator[r]{\sum_{(a,b) \in \mathcal{A} \times  \mathcal{B}}} \pi_a(s)\cdot   \sigma_b(s)\cdot   p(s,a,b,s').
 \end{align*}
Then, the Bellman evaluation operator $\mathfrak{T}^{\bm{\pi}, \bm{\sigma} } \colon \Real^S \to  \Real^S$ is defined for each $\bm{v}\in \Real^n$ and $s\in \mathcal{S}$ as
\begin{align} \label{eq:Bellman-Game}
\mathfrak{T}^{\bm{\pi}, \bm{\sigma}}_s \bm{v}
\; :=\; 
 r_s^{\bm{\pi}, \bm{\sigma}} + \gamma \cdot  \bm{P}^{\bm{\pi}, \bm{\sigma}}_s \bm{v}  ~.
\end{align}
For all operators, we use the shorthand $\mathfrak{T} \bm{v}_s := (\mathfrak{T} \bm{v})_s$. The Bellman equilibrium operator $\mathfrak{T}\opt \colon \Real^S \to \Real^S$ is defined as $\mathfrak{T}\opt_s \bm{v} := \max_{\bm{\pi} \in \Pi} \min_{\sigma \in \Sigma} \mathfrak{T}_s^{\bm{\pi}, \bm{\sigma}} \bm{v}$. The Bellman policy operator $\mathfrak{B}\opt\colon \Real^S \to  2^{\Pi \times \Sigma}$ computes the saddle point policies and is defined as
\begin{equation} \label{eq:bellman-policy} 
\mathfrak{B}\opt \bm{v}
  \; :=\; 
  \saddle( (\bm{\pi}, \bm{\sigma}) \mapsto \mathfrak{T}^{\bm{\pi}, \bm{\sigma}} \bm{v}),
  % \qquad
  % \mathfrak{B}\opt \bm{v} := \bigtimes_{s\in \mathcal{S}} \mathfrak{B}\opt_s \bm{v}.
\end{equation}
where the partial order on the value functions is defined as $\bm{u} \le \bm{v} \Leftrightarrow u_s \le v_s, \forall s\in \mathcal{S}$.

Bellman operators can be used to compute both $\bm{v}^{\bm{\pi}, \bm{\sigma}}$ for any $(\bm{\pi}, \bm{\sigma}) \in \Pi \times \Sigma$, as well as $\bm{v}^\star$. These value functions defined in~\eqref{eq:value-policy} and~\eqref{eq:value-optimal} are the \emph{unique} solutions for each $\bm{\pi}\in \Pi$ and $\bm{\sigma} \in \Sigma $ to, respectively~\cite[corollary~10.1]{Kallenberg2022},
\begin{equation*}
\bm{v}^{\bm{\pi}, \bm{\sigma}} = \mathfrak{T}^{\bm{\pi}, \bm{\sigma}} \bm{v}^{\bm{\pi}, \bm{\sigma}},
  \qquad
  \bm{v\opt} = \mathfrak{T}\opt  \bm{v\opt}.
\end{equation*}
The Bellman operators $\mathfrak{T}^{\bm{\pi}, \bm{\sigma}}$ and $\mathfrak{T}\opt$ are monotone and $\gamma$-contractive in the $L_{\infty}$ norm~\cite[theorem~10.5]{Kallenberg2022}. Because solutions to saddle points can be computed only approximately in polynomial time, we also define \emph{approximate Bellman equilibrium operator} $\mathfrak{T}^{\delta}\colon \Real^S \to \Real^S$ which satisfies that
\begin{equation} \label{eq:bellman-delta}
  \| \mathfrak{T}^{\delta} \bm{v} - \mathfrak{T}\opt \bm{v} \|_{\infty} \le \delta, \qquad
  \forall \bm{v}\in \Real^S,
\end{equation}
and the \emph{approximate Bellman policy operator} $\mathfrak{B}^\delta \colon \Real^S \to 2^{\Pi \times \Sigma}$ which satisfies
\[
\mathfrak{B}^\delta \bm{v} \subseteq\saddle_\delta( (\bm{\pi}, \bm{\sigma}) \mapsto \mathfrak{T}^{\bm{\pi}, \bm{\sigma}} \bm{v}).
\]
The well-known \emph{value iteration} is the simplest method for computing $\bm{v}\opt$ iteratively as $\bm{v}^{k+1} = \mathfrak{T}\opt \bm{v}^k$, where it is well-known that $\lim_{k\to \infty} \bm{v}^k = \bm{v}\opt$. It's worth noting that $\mathfrak{T}\opt$ is typically replaced with $\mathfrak{T}^\delta$ which has similar convergence properties.

% \begin{align*}
%     \mathfrak{T}^\delta\bm{v} = \left\{ \bm{x} : \|\bm{x} - \mathfrak{T}\opt\bm{v} \|_\infty \leq \delta \right\}.
% \end{align*}

% \kb{Operators are not monotone. Consider $\bm{r^{\pi,\sigma}} = \begin{bmatrix}
%     1 \\
%     -1
% \end{bmatrix}$ and $\bm{P^{\pi,\sigma}} = \begin{bmatrix}
%     0 & 1 \\
%     1 & 0
% \end{bmatrix}$. marek: that example does not violate monotonicity: $a \ge  b \implies  T a \ge  T b$ }

Computing the exact equilibrium is often unnecessary. Instead, it may be sufficient to compute an $\epsilon$-equilibrium for a sufficiently small $\epsilon$. To evaluate how close the value function is to the equilibrium, it is convenient to define the \emph{Bellman residual} $\psi_p\colon \Real^S\to \Real$ as
\begin{equation*}
\psi_p(\bm{v})
\; :=\;
\| \mathfrak{T}\opt \bm{v} - \bm{v} \|_p,
\qquad
p \in \left\{ 1, 2, \infty  \right\},
\end{equation*}
and the \emph{approximate Bellman residual} as
\begin{equation*}
    \psi^\delta_p(\bm{v}) \; := \; \| \mathfrak{T}^\delta \bm{v} - \bm{v} \|_p.
\end{equation*}

The following proposition shows that we can obtain $\epsilon$-equilibrium policies from a value function that approximates the equilibrium value function.
\begin{proposition} \label{prop:value-approximation-error}
For each $\bm{v} \in \Real^S$:
\[
  \emptyset \neq 
  \mathfrak{B}\opt \bm{v}
  \; \subseteq\;  \saddle_{\epsilon}(\rhog),
  \quad \text{where} \quad
 \epsilon = \frac{2 \gamma}{1-\gamma} \psi_{\infty}(\bm{v}).
\]
\end{proposition}
The proof, which we include in the appendix for the sake of completeness, follows standard arguments; see, for example,~\cite[theorem~10.11]{Kallenberg2022}. We note that the bound in \cref{prop:value-approximation-error} is tighter than the bounds given, for example, in~\citet[corollary~A.4]{Ho2021} and~\citet[theorems 3.1, 3.2]{Williams1993a}.

\subsection{Robust MDPs}
RMDPs generalize MDPs to allow for adversarial perturbations to the transition probabilities. We consider s-rectangular RMDPs $(\mathcal{S}, \mathcal{A}, r, \mathcal{P}, s_0)$ where $\mathcal{S}$ and $\mathcal{A}$ are the finite non-empty sets of states and actions, respectively~\cite{Wiesemann2013,Ho2021}, and $r\colon  \mathcal{S} \times \mathcal{A}  \to [-r_{\max},r_{\max}]$ is the reward function. The ambiguity set $\mathcal{P} := (\mathcal{P}_s)_{s\in \mathcal{S}}$, where $\mathcal{P}_s \subseteq \Delta^\mathcal{S}$ is compact and non-empty for each $s\in \mathcal{S}$, and determines the range of possible adversarial transition probability functions. Finally, the initial state is $s_0$. 

It is common to define the ambiguity sets in RMDPs as bounded norm-balls around a given nominal transition function $\bm{\bar{p}}\colon \mathcal{S} \times  \mathcal{A} \to  \Delta^S$, such as~\cite{Ho2021,Ho2022,Behzadian2021,Behzadian2021a} 
\begin{equation*}
\mathcal{P}_s
:=
\left\{ \bm{p} \in (\Delta^{\mathcal{S}})^{\mathcal{A}}  \mid 
\sum_{a\in\mathcal{A}}  \| \bm{p}(a) - \bm{\bar{p}}(s,a) \| \leq \xi_s \right\}, 
\end{equation*}
for some norm $\| \cdot  \|$ and $\xi_s \ge 0,\, s\in \mathcal{S}$. In this work, we focus on ambiguity sets defined by the $L_1$-norm.

As with Markov games as defined above, we seek to compute a stationary policy $\pi\in \Pi$ that maximizes the expected $\gamma$-discounted infinite-horizon robust return: 
\begin{equation} \label{eq:robust-objective}
\begin{gathered}
\max_{\bm{\pi} \in \Pi} \min_{\bm{p} \in \mathcal{P}} \,
\rhor(\bm{\pi}, \bm{p}),\\
\rhor(\bm{\pi}, \bm{p}) :=
\E_{\bm{\pi}, \bm{p}}^{s_0} \left[ \sum_{t=0}^{\infty} \gamma^t r(\tilde{s}_t, \tilde{a}_t)\right],
\end{gathered}
\end{equation}
where $\gamma \in (0,1)$. We emphasize that for each $\bm{\pi} \in \Pi$, the domain of $\rhor(\bm{\pi}, \cdot)$ is the set of feasible transition probabilities $\mathcal{P}$, rather than the set of all transition probabilities. An optimal policy $\bm{\pi\opt}$ in~\eqref{eq:robust-objective} exists and can be computed by robust value or policy iteration~\cite{Wiesemann2013,Iyengar2005}. 

The following proposition shows that the concept of approximate optimality for RMDPs is closely related to the concept of approximate saddle points in games.
\begin{proposition} \label{prop:rmdp-bound}
  Suppose that $(\bm{\hat{\pi}}, \bm{\hat{p}}) \in \saddle_{\epsilon}(\rhor)$ for some $\epsilon \ge 0$. Then $\bm{\hat{\pi}}$ is \emph{$2\epsilon$-robust optimal} in the sense that
  \[
    \min_{\bm{p}\in \mathcal{P}}  \rhor(\bm{\hat{\pi}}, \bm{p})
    \; \ge\;  
   \min_{\bm{p}\in \mathcal{P}}  \rhor(\bm{\pi\opt}, \bm{p}) - 2\cdot \epsilon. 
 \]
\end{proposition}

For RMDPs, value functions and Bellman operators are defined analogously to how they are defined for MGs. For more detail, please see the appendix. In the remainder of the paper, we describe the algorithms for MGs which generalize to RMDPs. 

Computationally, the main difference between RMDPs and MGs is in computing the Bellman operator. For most common ambiguity sets $\mathcal{P}$, the robust Bellman operator can be implemented by solving a convex optimization problem~\cite{Wiesemann2013}. For the $L_1$-bound ambiguity sets, the robust Bellman operator can be implemented by solving a linear program~\cite[appendix~C]{Ho2021}. Significantly more efficient methods exist for ambiguity sets bounded by norms and $\varphi$-divergences~\cite{Ho2021,Ho2022,Behzadian2021a}. 

\section{Filar-Tolwinski Algorithm\\May Not Converge} \label{sec:filar-tolw-suboptimal}

\citet{PAI1969} proposed one of the first alternatives to value iteration~\cite{Shapley1953} for solving MGs. This algorithm, which we refer to as the PAI algorithm, can be viewed as applying Newton's method to the problem of finding a zero of $\psi_2(\bm{v})^2$. While PAI is known to converge to the optimal value function $\bm{v}\opt$ under certain restrictive conditions~\cite[theorem~5]{PAI1969}, it is also known to not converge at all for certain MGs~\cite{VanderWal1978}. Filar and Tolwinski~\cite{Filar1991} proposed a modified Newton method intended to fix this convergence issue. In this section, we provide a counterexample to \citet[theorem~3.3]{Filar1991}, where it is claimed that the modified Newton method converges from some constant initial vector to $\bm{v}\opt$. The Filar-Tolwinski (FT) algorithm is described in \cref{alg:FilarTolwinski}.

\begin{algorithm}
    \caption{Filar-Tolwinski (FT) Algorithm}
    \label{alg:FilarTolwinski}
    \KwIn{Initial value $\bm{v}^0$, tolerance $\epsilon$, backtracking line search coefficients $\beta\in (0,1), \delta \in (0,1)$}
    \KwOut{$(\bm{\pi}, \bm{\sigma}) \in \saddle_{\epsilon}(\rhog)$}
    $k \gets 0$\;
    \Repeat{$\frac{2 \gamma}{1-\gamma} \cdot \psi_\infty(\bm{v}^k) \le \epsilon$}
    {
        $k \gets k + 1$\;
        Select $(\bm{\pi}^k,\bm{\sigma}^k)\in \mathfrak{B}\opt \bm{v}^{k-1}$\;
        $\bm{d}^k \gets (\bm{I} - \gamma\bm{P}^{\bm{\pi}^k,\bm{\sigma}^k})^{-1}\bm{r}^{\bm{\pi}^k,\bm{\sigma}^k} - \bm{v}^{k-1}$\;
        \tcp{Line search, Armijo's rule:}
        \tcp{$\nabla \psi_2(\bm{v})^2 =  2 (\gamma\bm{P}^{\bm{\pi}^k,\bm{\sigma}^k} - \bm{I})\tr(\mathfrak{T}\opt\bm{v} - \bm{v})$}
        \label{ln:ft-backtrack}
        $i_k \gets \min \{ i \in \mathbb{N} \mid  \psi_2(\bm{v}^{k-1} + \beta^i\bm{d}^k)^2 \le$ \\
          $\qquad \le  \psi_2(\bm{v}^{k-1})^2 + \delta \beta^i \cdot (\bm{d}^k)\tr \nabla \psi_2(\bm{v}^{k-1})^2 \} $ \;
        \label{ln:ft-stepsize}
        $\bm{v}^k \gets \bm{v}^{k-1} + \beta^{i_k} \cdot \bm{d}^k $\;
    }
    \KwRet $(\bm{v}^k, \bm{\pi}^k, \bm{\sigma}^k)$\;
\end{algorithm}

To derive the FT algorithm, one interprets PAI as the pure Newton's method for solving $\min_{\bm{v}\in \Real^S}  \psi_2(\bm{v})^2$~\cite{Filar1991,Filar1996}. Recall that the pure Newton's method direction $\bm{d}^k$ in iteration $k\in \Nats$ is
\begin{align*}
  \bm{d}^k &:= -(\nabla^2 \psi_2(\bm{v}^{k-1})^2)^{-1} \nabla \psi_2(\bm{v}^{k-1})^2 \\
  &= \left(\bm{I} - \gamma  \bm{P}^{\bm{\pi}^k,\bm{\sigma}^k} \right)^{-1} \left(\mathfrak{T}\opt\bm{v}^{k-1} - \bm{v}^{k-1} \right).
\end{align*}

FT's insight is to replace the pure Newton's step size of $1$ in PAI with a backtracking line search. Setting the step size in Line~\ref{ln:ft-stepsize} in \cref{alg:FilarTolwinski} to $i_k = 0$  recovers PAI exactly. The use of Armijo's rule in determining FT ensures that the objective function $\psi_2(\bm{v})^2$ decreases in every step. Since $\psi_2(\bm{v}\opt) = 0$ is the unique global minimum of $\bm{v} \mapsto \psi_2(\bm{v})^2$ and each of FT's iterations decreases the objective function, FT cannot cycle and does not terminate until reaching the optimal value function $\bm{v}\opt$.

Theorem~3.3 in~\cite{Filar1991} states that \cref{alg:FilarTolwinski} is guaranteed to converge to the optimal value function. However, there is a gap in the proof. In particular, while each step of the iteration reduces the value function, it is not guaranteed that a step size satisfying Armijo's rule exists. Since the gradient of $\bm{v} \mapsto \psi_2(\bm{v})^2$ may be discontinuous, it is possible that no $i$ in Line~\ref{ln:ft-backtrack} in \cref{alg:FilarTolwinski} satisfies the inequality; leading to an infinite loop in the search for the step size. We construct a simple MDP example demonstrating this behavior to show that this can happen.

\begin{figure}
    \centering
    \includegraphics[width=\linewidth]{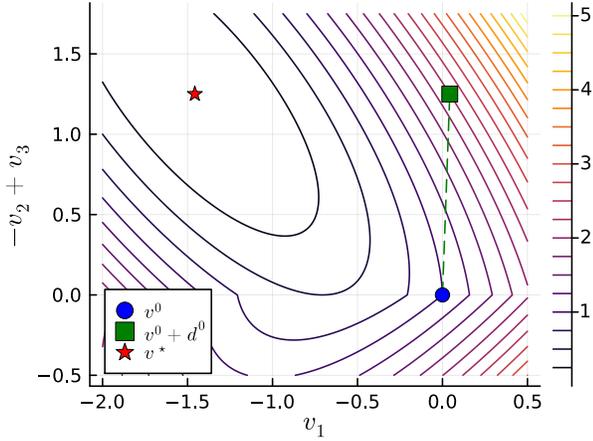}
    \caption{Plot of $\psi_2(\bm{v})^2$ projected onto the plane that spans the initial value function, optimal value function, and the step direction.}
    \label{fig:counter_plot}
\end{figure}

\begin{example} \label{exm:local-minimum}  
Consider a MG with $\mathcal{S} = \left\{ s_1,s_2,s_3 \right\}$, $\mathcal{A} = \left\{ a_1 \right\}$, and $\mathcal{B} = \left\{ b_1, b_2 \right\}$. The transition probabilities and rewards are defined in \cref{tbl:counter-game}. The columns represent actions. When only one column exists in a state, all actions behave identically. The top row of each cell represents the reward associated with the action, and the bottom row represents the transition probability function for that state and action. The discount factor is $\gamma = 0.6$. 
\end{example}

\begin{figure}
  \centering
  \hfill
\begin{tabular}{|c|c|}
\hline
  \multicolumn{2}{|c|}{$s_1$} \\
  \hline
  $b_1$ & $b_2$ \\
  \hline
  $-\sqrt{2}/2$ & $-\sqrt{2}/2$ \\
$[0,0,1]$ & $[0,1,0]$ \\
\hline
\end{tabular}
\hfill
\begin{tabular}{|c|}
\hline
  $s_2$ \\
  \hline
  $b_1$ \\
\hline
    $-1/2$ \\ 
 $[0,1,0]$ \\
\hline
\end{tabular}
\hfill
\begin{tabular}{|c|}
\hline
  $s_3$ \\
  \hline
  $b_1$ \\
\hline
    $1/2$ \\ 
 $[0,0,1]$ \\
\hline
\end{tabular}
\hfill
\caption{Rewards and transition probabilities of the Markov game for states $s_1, s_2, s_3$ from \cref{exm:local-minimum}.} \label{tbl:counter-game}
\end{figure}

The following theorem formally states that \cref{exm:local-minimum} is a counterexample to the optimality of FT.
\begin{theorem} \label{thm:counter-example}
FT in \cref{alg:FilarTolwinski} initialized to $\bm{v}^0 = \bm{0}$ and applied to the MG in \cref{exm:local-minimum}  visits only suboptimal policies and never terminates.
\end{theorem}

We note that \citet[theorem~3.3]{Filar1991} assumes that FT is initialized to a constant value determined by the maximum reward instead of a zero vector. However, the algorithm makes no progress even with such initialization as shown in the appendix. 

We now discuss the gap in the proof of convergence in~\citet[theorem~3.3]{Filar1991} which \cref{thm:counter-example} contradicts. As noted in~\citet[theorem~2.1]{Filar1991} the function $\bm{v} \mapsto \psi_2(\bm{v})^2$ is differentiable almost everywhere. However, because the function is not differentiable everywhere, Armijo's rule fails to find a positive step size. \Cref{exm:local-minimum} initializes FT in exactly a point of non-differentiability. One attempt to circumvent this problem would be to argue that the probability of being at a point of non-differentiability is zero. However, it may be possible to modify our example so that the initialization does not happen at a point of non-differentiability. Yet, the line search method will take increasingly smaller steps, such that it approaches the point of non-differentiability without ever passing it. Because it is unclear how one may rectify the non-differentiability to ensure Newton's method's convergence, we propose an alternative approach in the following section.

\section{RCPI: Residual Conditioned Policy Iteration} \label{sec:rcpi:-resid-corr}

In this section, we propose and analyze a new algorithm, RCPI, for solving MGs and RMDPs. RCPI builds on the strengths of PAI and FT but with convergence guarantees. Our theoretical analysis demonstrates that RCPI is guaranteed to converge to the optimal value function at a rate that at least matches that of value iteration.

\begin{algorithm}
  \caption{RCPI: Residual Conditioned PI}
    \label{alg:RCPI}
    \KwIn{Initial value $\bm{v}^0$, tolerance $\epsilon$, backup tolerance $\delta < \epsilon \cdot  \frac{(1-\gamma)^2}{2\gamma (3+\gamma)}$, max recovery steps $m\in \Nats$}
    \KwOut{$(\bm{\pi}, \bm{\sigma}) \in \saddle_{\epsilon}(\rhog)$}
    $k \gets 0$\;
    \Repeat{$\frac{2\gamma}{1-\gamma}\left(\psi^\delta_{\infty}(\bm{v}^k) + \delta\right)  \le \epsilon$}
    {
        $k \gets k + 1$ \label{ln:CounterIncrement} \label{ln:IterationBeginning}\;
        Select $(\bm{\pi}^k,\bm{\sigma}^k) \in \mathfrak{B}^\delta \bm{v}^{k-1}$\label{ln:PolicySelection}\;
        $\bm{u}^{k,0} \gets (\bm{I} - \gamma\bm{P}^{\bm{\pi}^k,\bm{\sigma}^k})^{-1}\bm{r}^{\bm{\pi}^k,\bm{\sigma}^k}$\label{ln:PolicyEvaluation}\;
        \lIf{$\gamma^{m-1}\psi^\delta_{\infty}(\bm{u}^{k,0}) + \frac{2(1 + \gamma)\delta}{1-\gamma} >  \psi^\delta_{\infty}(\bm{v}^{k-1})$ \label{ln:InitialResidualCheck}} {
            $\bm{v}^k \gets \mathfrak{T}^\delta \bm{v}^{k-1}$
        }
        \Else {
          \label{ln:FixingCheck}
          $l \gets 0$\;
          \lWhile{$\psi^\delta_{\infty}(\bm{u}^{k,l}) > \gamma \psi^\delta_{\infty}(\bm{v}^{k-1}) + 2(1+\gamma)\delta$} {
                \label{ln:FixingBackups}
                 $\bm{u}^{k,l+1} \gets \mathfrak{T}^\delta  \bm{u}^{k,l} $; $l \gets l + 1$
            }
            $\bm{v}^k \gets \bm{u}^{k,l}$ \label{ln:IterationEnding}\;
        }
    } \label{ln:TerminationCheck}
    \KwRet $(\bm{v}^k, \bm{\pi}^k, \bm{\sigma}^k)$\;
\end{algorithm}

RCPI, summarized in \cref{alg:RCPI}, can be viewed as a direct modification of FT in \cref{alg:FilarTolwinski}. The first two steps of RCPI's iteration are identical to FT. First, RCPI jointly updates both the primary and adversarial policies to be greedy with respect to the current value function. Second, RCPI evaluates the value function for the updated policies. Simply adopting this value function would lead to PAI (see the appendix), which is prone to getting stuck in infinite cycles~\cite{VanderWal1978}. Such infinite cycles must involve steps that do not decrease the residual. RCPI detects when the residual does not decrease sufficiently and reverts to a value function update to guarantee its reduction. As a result, RCPI will never cycle or terminate before reaching the optimal value function. 

RCPI guarantees convergence to the optimal value function as follows. Each iteration of the outer loop guarantees that the residual of the incumbent value function decreases at least by the factor $\gamma$. The parameter $m$ determines how reduction is achieved. If the residual of the proposed value function can be reduced in at most $m$ steps of value iteration, then the Bellman operator is applied until the reduction is achieved. Otherwise, the proposed value function is discarded and replaced by a plain value iteration update.  

We now turn to the proof of RCPI's correctness and computation complexity. First, we need to discuss the worst-case runtime of the Bellman backups. For s-rectangular $L_1$ robust MDPs the runtime of computing $\mathfrak{T}^\delta\bm{v}$ and $\mathfrak{B}^\delta\bm{v}$ is~\cite{Ho2021}
\begin{align*}
    T_{\mathrm{R}} = O \left( S^{4.5}A^{4.5} \right)~,
\end{align*}
and, for MGs, it is given by \cref{lem:GameBackupRuntime}.
\begin{proposition} \label{lem:GameBackupRuntime}
The runtime $T_{\mathrm{G}}$ of computing $\mathfrak{T}^\delta\bm{v}$ and $\mathfrak{B}^\delta\bm{v}$ for a Markov game satisfies that
\begin{align} \label{eq:MatrixGameRuntime}
T_{\mathrm{G}}  = O\left( S^2AB + S(A + B)^{1.5}(A)^2\log(\delta^{-1}) \right) ,
\end{align}
where, without loss of generality, $A \ge B$.
%If $A < B$, the rewards can be negated by swapping the actions available to the maximizing and minimizing agents, and exchanging the values of $A$ and $B$ without affecting the optimal policies.
\end{proposition}

We are now ready to state the central claim of this section, which proves the correctness and computational complexity of RCPI. 
\begin{theorem} \label{thm:RCPIRuntime}
  Suppose that $\gamma > 0$, and $\epsilon > 0$ satisfies that
  \[
    \epsilon
    \; >\;  \frac{2(1+\gamma)\delta}{(1-\gamma)^2}
    \; >\;  0.
  \]
  for $\delta$ in~\eqref{eq:bellman-delta}. Then \cref{alg:RCPI} returns $(\bm{\pi},\bm{\sigma}) \in \saddle_\epsilon(\rho)$ in $O\left( Z \left(T \cdot ( 1 + m ) + S^2AB + S^3\right) \right)$ operations where
\begin{align} \label{eq:RCPIRuntime}
  Z :=
    \left\lceil
    \frac{\log\left(\frac{1-\gamma}{2\gamma}\epsilon  -
        \frac{3+\gamma}{1-\gamma}\delta\right) -
        \log(r_{\max}+ \delta)}{\log(\gamma)} 
    \right\rceil,
\end{align}
and $T \in \left\{ T_{\mathrm{R}}, T_{\mathrm{G}} \right\}$ is the complexity of computing $\mathfrak{T}^\delta\bm{v}$ and $\mathfrak{B}^\delta\bm{v}$ for RMDP or MG, respectively.
\end{theorem}
The proof of \cref{thm:RCPIRuntime} follows standard contraction arguments and is deferred to the appendix. The main argument relies on the following lemma, which bounds the computational time and establishes the contraction property of each iteration of RCPI.
\begin{lemma} \label{lem:IterationContraction}
Each loop of \cref{alg:RCPI} (Lines~\ref{ln:IterationBeginning}--\ref{ln:IterationEnding})  runs in
\begin{align} \label{eq:IterationRunTime}
    O\left( \left( 1 + m \right) T  + S^2AB + S^3\right)
\end{align}
operations for $T \in  \left\{ T_{\mathrm{R}}, T_{\mathrm{G}} \right\}$
and $(\bm{v}^k)_{k \in \Nats}$ satisfies that
\begin{align} \label{eq:BellmanContraction}
\psi^\delta_\infty(\bm{v}^{k+1})
\; \le\;
\gamma\cdot \psi^\delta_\infty(\bm{v}^k) + 2\cdot (1+\gamma)\cdot \delta~.
\end{align}
\end{lemma}

Note that the number of Bellman backups required for RCPI to find  $(\bm{\pi},\bm{\sigma}) \in \saddle_\epsilon(\rho)$ shares the same upper bound as robust VI for both games and robust MDPs if the maximum number of recovery steps $m$ is set to 0. RCPI's main attraction is that it can leverage its exact policy evaluation to aid in finding an optimal solution, while ignoring or correcting it when issues arise. This gives it speeds that are close to, if not faster than, PAI when solving problems in practice. As a result, the worst-case time complexity of RCPI is no worse than value iteration, but it offers significant possible speedup due to the policy evaluation step. 

\section{Numerical Results} \label{sec:numerical-results}

To evaluate the effectiveness of RCPI a series of examples were solved using a range of algorithms including PAI, Hoffman Karp (HK), Filar Tolwinski's algorithm (FT), robust value iteration (VI), a variation of Hoffman Karp (PPI)~\cite{Ho2021}, a variation on PAI (WS)~\cite{Winnicki2023}, and our algorithm (RCPI). To simplify comparison, RCPI's hyperparameter $m$ was either set to 0, producing a method that never fixed its evaluation step, called RCPI$_0$, or $m$ was left unbounded, making a method that always fixed its evaluation step, called RCPI$_\infty$. The full source code for the algorithms and domains is available at \url{https://github.com/keithbadger/Fast-Policy-Iteration-for-Markov-Games-and-Robust-MDPs}.

For each domain, there was a smaller set of problems solved with $\gamma$ set to $0.5$, $0.75$, $0.9$, and $0.99$, and a larger set of problems where we set $\gamma$ to $0.9$, $0.99$, and $0.999$. We use the smaller problems in order to evaluate the slower algorithms in a reasonable time. The larger problems can only be solved by the faster methods within our time limits. We allocated a larger time budget to VI to establish a reference.  

\begin{table*}
    \centering
    \begin{tabular}{lrrrrrrrr}
        \toprule
        &\multicolumn{2}{c}{Markov Games} &\multicolumn{2}{c}{Inventory} &\multicolumn{2}{c}{Gambler's Ruin} &\multicolumn{2}{c}{Gridworld}\\
        \cmidrule(lr){2-3} \cmidrule(lr){4-5} \cmidrule(lr){6-7} \cmidrule(lr){8-9} 
        \multicolumn{1}{c}{Algorithm} & \multicolumn{1}{c}{Small} & \multicolumn{1}{c}{Large} & \multicolumn{1}{c}{Small} & \multicolumn{1}{c}{Large} & \multicolumn{1}{c}{Small} & \multicolumn{1}{c}{Large} & \multicolumn{1}{c}{Small} & \multicolumn{1}{c}{Large}\\
        \midrule
        RCPI$_\infty$ & 0.3 & 2.3 & 0.1 & 84.8 & 4.7 & 54.3 & 1.5 & 23.7\\
        RCPI$_0$ & 0.3 & 2.3 & 0.2 & 87.8 & 5.1 & 54.4 & 1.4 & 23.7\\
        VI & 3.4 & 253.0 & 2.4 & 23629.6 & 8.7 & 106.6 & 6.9 & 145.2\\
        PAI & 0.3 & 2.3 & 0.1 & 87.2 & 4.8 & 54.3 & 1.4 & 23.6\\
        FT & 0.3 & 2.4 & 0.2 & 86.9 & 5.6 & 77.1 & 1.4 & 23.3\\
        HK & 0.5 & * & 0.4 & * & 14.4 & * & 5.7 & *\\
        WS & 1.0 & * & 0.8 & * & 5.2 & * & 3.2 & *\\
        PPI & 0.6 & * & 0.4 & * & 10.6 & * & 4.9 & *\\
        \bottomrule
    \end{tabular}
    \caption{The median runtime of each algorithm's in seconds for the small and large problem sets of every domain.}
    \label{tab:Summary}
\end{table*}

\begin{figure}
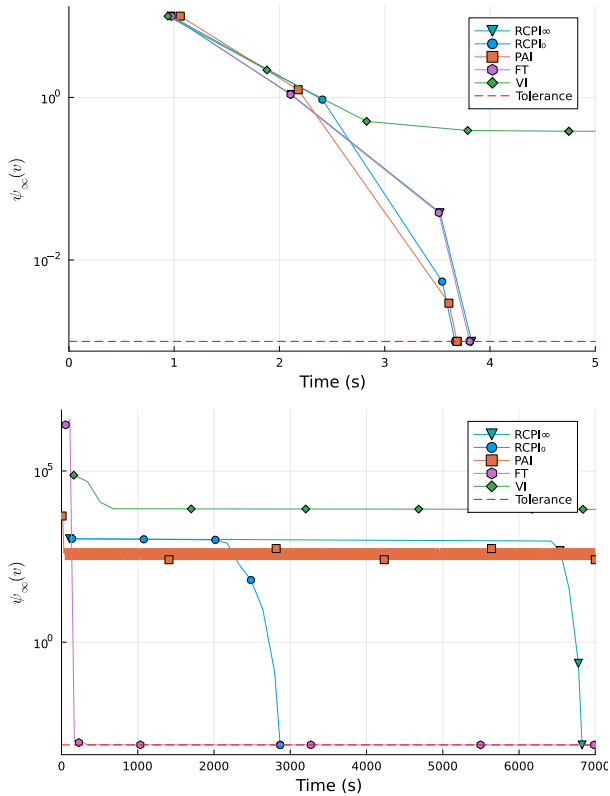

    \begin{subfigure}{.45\textwidth}
        \includegraphics[width=1\linewidth]{figures/games\_large.pdf}
    \end{subfigure}
    \begin{subfigure}{.45\textwidth}
        \includegraphics[width=1\linewidth]{figures/inv\_large.pdf}
        \label{fig:InventoryInPaper}
    \end{subfigure}
    \caption{The Bellman residual of each algorithm's value function plotted as a function of time for the \emph{large Markov games} \emph{(top)} with 200 to 1000 states, and the \emph{large inventory problems} \emph{(bottom)} with 40 to 200 states.}
    \label{fig:BellmanResiduals}
\end{figure}

Examining \cref{tab:Summary}, VI is the slowest algorithm tested for games. Every other algorithm achieves a faster median solve time. The difference in solution time is due to VI exclusively using policy improvement steps to improve its estimated value function, whereas other methods incorporate policy evaluation steps, which are often more efficient.

WS is the closest method to VI conceptually, only adding a fixed number of policy evaluation backups in between policy improvement steps. The policy evaluation backups can significantly reduce the solve time of games, as shown in \cref{tab:Summary}, where all of WS's median runtimes are below those of VI.

The remaining methods use exact policy evaluation steps. HK and PPI's evaluations differ from the others, as they both evaluate the primary policy by holding it constant and optimizing the adversary. In contrast, the other methods hold both policies constant and find the stationary value function $\bm{v}$ which satisfies $\mathfrak{T}\bm{v} = \bm{v}$. Although HK and PPI's policy evaluation methods do improve upon VI's solve time, they are still cumbersome when compared to the closed-form methods of PAI, FT, RCPI$_0$, and RCPI$_\infty$, which evaluate both policies simultaneously. As a result, HK and PPI are the next slowest methods for games. 

The median runtimes of the closed-form methods were approximately the same across all domains. RCPI$_0$ and RCPI$_\infty$ were always within a few seconds of the fastest runtime. The worst-case runtimes of the closed-form methods from \cref{fig:BellmanResiduals} were similarly close, with all of their Bellman residual curves indicating a super-linear convergence rate.

Several domains were tested for robust MDPs including gamblers ruin~\cite{Kallenberg2022}, gridworld~\cite[section 6.5]{Sutton1998}, and inventory management~\cite[section 3.2]{Puterman2005}. From \cref{tab:Summary}, each algorithm maintained the same relative performance from Markov games, except that WS and VI did comparatively better in the gambler's ruin. WS and VI did well because the optimal betting scheme in the non-robust version of gambler's ruin is to bet \$1 if the win rate is greater than $50\%$ and to bet all money otherwise. The gambler repeats this action until reaching the maximum capital, when it obtains the reward. The result is a singular optimal state trajectory following the current state. The reward from obtaining the maximum capital does not affect the current state's policy until there is a Bellman backup for every state in the optimal trajectory following it. This type of domain favors methods with inexpensive evaluations, as only states that reward has been reached via policy improvement will be worth evaluating. Time spent doing exact evaluations for the other states does not provide any benefit.

In \cref{tab:Summary}, the closed evaluation methods have the lowest median runtimes for the small and large inventory problem sets. By examining the inventory problems in \cref{fig:BellmanResiduals}, PAI stops converging around $10^2$, where it becomes stuck cycling between suboptimal value function estimates, as described in~\cite{VanderWal1978}. At points, FT's Bellman residual also increases, which comes from using $\psi_2(\bm{v})^2$ as an objective instead of $\psi_\infty(\bm{v})$. The larger number of available actions and transitions stemming from those actions makes policy improvement for inventory management more expensive than for the other domains. As a result, methods that evaluate their policies simultaneously benefit more heavily from limiting the number of policy improvement steps needed to converge.

\section{Conclusion}

Historically, solving Markov games and robust MDPs has involved choosing between a slow method that always converges, or a fast method that may never finish. Attempts have been made to provide a solution with both speed and convergence guarantees, such as the algorithm of Filar and Tolwinski, but such attempts have failed. RCPI is a simple solution which provides the best possible worst-case convergence rate, and empirically performs as fast if not faster than any method proposed before it.

It remains to be seen if there is a way to fix the algorithm of Filar and Tolwinski that keeps the Newtonian interpretation of PAI with $\psi_2(\bm{v})^2$ as the objective function. Such a solution could provide insight into how to deal with discontinuities caused by the min and max operations more generally. It is not known how to optimize the hyperparameter of RCPI $m$ which could reduce the level of knowledge required to use RCPI effectively.
%We suspect that there is no optimal selection of $m$ and that increasing it has the potential to reduce the expected solve time while increasing the runtime upperbound.

\section*{Acknowledgments} 
We thank the anonymous reviewers for their detailed reviews and thoughtful comments, which significantly improved the paper's clarity. This work was supported, in part, by NSF grants 2144601 and 2218063 and ONR grant N0001425GI01179.

\bibliography{paper}

% \iffalse
  
\appendix
\onecolumn

\section{Additional Material for Section~\ref{sec:background}}
\label{sec:addit-background}

\subsection{Proof of Proposition~\ref{prop:value-approximation-error}}

\begin{proof}[Proof of \cref{prop:value-approximation-error}]
Define for each $\bm{\pi}\in \Pi, \bm{\sigma}\in \Sigma$:
\[
\mathfrak{T}^{\star, \bm{\sigma}}_s \bm{v}
:=
\max_{\bm{\pi} \in \Pi} \mathfrak{T}^{\bm{\pi}, \bm{\sigma}}_s \bm{v},
\qquad
\mathfrak{T}^{\bm{\pi},\star}_s \bm{v}
:=
\min_{\bm{\sigma} \in \Sigma} \mathfrak{T}^{\bm{\pi}, \bm{\sigma}}_s \bm{v},
\quad \forall \bm{v}\in \Real^S.
\]
It is easy to see that these operators are monotone and $L_{\infty}$ contractions with a rate $\gamma$~\cite[theorem~10.5]{Kallenberg2022}. Also define:
\[
  \bm{v}^{\star, \bm{\sigma}} := \max_{\bm{\pi}\in \Pi} \bm{v}^{\bm{\pi}, \bm{\sigma}} ,
  \qquad
  \bm{v}^{\bm{\beta}, \star} := \min_{\bm{\sigma }\in \Sigma} \bm{v}^{\bm{\pi}, \bm{\sigma}} ,
\]
where the optimization is elementwise and exists~\cite[theorem~10.4]{Kallenberg2022}.

Then, fix a value function $\bm{v}\in \Real^S$ and let $(\bm{\hat{\pi}}, \bm{\hat{\sigma}}) \in \mathfrak{B}\opt\bm{v}$:
\begin{equation} \label{eq:proof-greedy}
  \mathfrak{T}^{\bm{\hat{\pi}}, \star} \bm{v}
  = 
  \mathfrak{T}^{\bm{\hat{\pi}}, \bm{\hat{\sigma}}} \bm{v}
  = 
  \mathfrak{T}^{\star, \bm{\hat{\sigma}}} \bm{v}.
\end{equation}
where the equalities hold because the maximum and minimum policies exist. 
Then, using~\eqref{eq:proof-greedy} and~\cite[proposition~2.1.1 (d),(e)]{Bertsekas2022}:
\[
  \begin{aligned}
  \| \bm{v}^{\star, \bm{\hat{\sigma}}} - \bm{v}^{\bm{\hat{\pi}}, \bm{\hat{\sigma}}} \|_{\infty}
  &=
  \| \bm{v}^{\star, \bm{\hat{\sigma}}} - \mathfrak{T}^{\star, \bm{\hat{\sigma}}} \bm{v}
    + \mathfrak{T}^{\bm{\hat{\pi}}, \bm{\hat{\sigma}}} \bm{v} - \bm{v}^{\bm{\hat{\pi}}, \bm{\hat{\sigma}}} \|_{\infty} \\
  &\le
  \| \bm{v}^{\star, \bm{\hat{\sigma}}} - \mathfrak{T}^{\star, \bm{\hat{\sigma}}} \bm{v} \|_{\infty} +
    \| \mathfrak{T}^{\bm{\hat{\pi}}, \bm{\hat{\sigma}}} \bm{v} - \bm{v}^{\bm{\hat{\pi}}, \bm{\hat{\sigma}}} \|_{\infty} \\
  &\le
    \epsilon.
  \end{aligned}
\]
An analogous argument gives us the following inequality:
\[
  \begin{aligned}
    \| \bm{v}^{\bm{\hat{\pi}}, \star} - \bm{v}^{\bm{\hat{\pi}}, \bm{\hat{\sigma}}} \|_{\infty}
  &=
  \| \bm{v}^{\bm{\hat{\pi}}, \star} - \mathfrak{T}^{\bm{\hat{\pi}}, \star} \bm{v}
    + \mathfrak{T}^{\bm{\hat{\pi}}, \bm{\hat{\sigma}}} \bm{v} - \bm{v}^{\bm{\hat{\pi}}, \bm{\hat{\sigma}}} \|_{\infty} \\
  &\le
  \| \bm{v}^{\star, \bm{\hat{\sigma}}} - \mathfrak{T}^{\star, \bm{\hat{\sigma}}} \bm{v} \|_{\infty} +
    \| \mathfrak{T}^{\bm{\hat{\pi}}, \bm{\hat{\sigma}}} \bm{v} - \bm{v}^{\bm{\hat{\pi}}, \bm{\hat{\sigma}}} \|_{\infty} \\
  &\le
    \epsilon.
  \end{aligned}
\]
Therefore, we can conclude from the two inequalities above that
\[
    v_{s_0}^{\bm{\hat{\pi}}, \star} - \epsilon
  \le
  v_{s_0}^{\bm{\hat{\pi}}, \bm{\hat{\sigma}}}
  \le 
   v_{s_0}^{\star, \bm{\hat{\sigma}}} + \epsilon ,
\]
which is the definition of $(\bm{\hat{\pi}}, \bm{\hat{\sigma}}) \in \saddle_{\epsilon}(\rhog)$. The set $\mathfrak{B}\opt \bm{v}$ is non-empty by the construction of $\mathfrak{T}^{\star, \bm{\sigma}}$ and $\mathfrak{T}^{\bm{\pi}, \star}$.
\end{proof}

\subsection{Definitions for Robust MDPs}
\label{sec:defin-robust-mdps}
In this section, we summarize the pertinent definitions of value functions and Bellman operators in Robust MDPs. The \emph{value function} $\bm{v}^{\bm{\pi}, \bm{p}} \in \Real^S$ is defined as for each $\bm{\pi}\in \Pi$ and $\bm{p}\in \mathcal{P}$ as~\cite{Wiesemann2013, Ho2022}:
\begin{equation} \label{eq:value-policy-robust}
v^{\bm{\pi},\bm{p}}_s
  \; :=\; 
  \E_{\bm{\pi}, \bm{p}}^s
  \left[ \sum_{t=0}^{\infty} \gamma^t r(\tilde{s}_t, \tilde{a}_t)\right],
  \quad
  \forall s\in \mathcal{S}.
\end{equation}
The superscripts and subscripts of $\E^s_{\bm{\pi}, \bm{p}}$ indicate that the probability measure is chosen such that $\tilde{s}_0 = s$ and that $\tilde{a}_t \sim \bm{\pi}(\tilde{s}_t)$, and $\tilde{s}_{t+1} \sim \bm{p}(\tilde{s}_t, \tilde{a}_t)$ for all $t\in \Nats$.  The \emph{optimal robust value function} $\bm{v\opt} \in \Real^S$ is defined as the saddle point over the policies:
\begin{equation} \label{eq:value-optimal-robust}
  v\opt_s
  \; :=\;
  \max_{\bm{\pi}\in \Pi } \min_{\bm{p} \in \mathcal{P} } v^{\bm{\pi}, \bm{p} }_s, \qquad \forall s\in \mathcal{S}.
\end{equation}

Next, we describe the \emph{Bellman operator} for RMDPs. For each $\bm{\pi}\in \Pi$ and $\bm{p}\in \mathcal{P}$, we define the reward vector $\bm{r}^{\bm{\pi}}\in \Real^S$ and a transition matrix $\bm{P}^{\bm{\pi},\bm{p}} \in \Real_{+}^{S \times  S}$ as
 \begin{align*}
   P^{\bm{\pi},\bm{p}}_{s,s'} := \sum_{a \in \mathcal{A}} \pi_a(s)\cdot   p_{s'}(s,a),
   \qquad
     r^{\bm{\pi}}_s := \sum_{a \in \mathcal{A} } \pi_a(s)  \cdot r(s,a).
 \end{align*}
Then, the \emph{Bellman evaluation operator} $\mathfrak{T}^{\bm{\pi}, \bm{p} } \colon \Real^S \to  \Real^S$ is defined for each $\bm{v}\in \Real^n$ and $s\in \mathcal{S}$ as
\begin{align} \label{eq:Bellman-Game-robust}
\mathfrak{T}^{\bm{\pi}, \bm{p}}_s \bm{v}
\; :=\; 
r_s^{\bm{\pi}} + \gamma \cdot \smashoperator[r]{\sum_{a \in \mathcal{A}}}   \pi_a(s) \cdot  \bm{p}(s,a)\tr \bm{v} , \quad \forall s \in \mathcal{S} ~.
\end{align}
The \emph{Bellman equilibrium operator} $\mathfrak{T}\opt \colon \Real^S \to \Real^S$ is defined as $\mathfrak{T}\opt_s \bm{v} := \max_{\bm{\pi} \in \Pi} \min_{p \in \mathcal{P}} \mathfrak{T}^{\bm{\pi}, \bm{p}}_s \bm{v}$. The \emph{Bellman policy operator} $\mathfrak{B}\opt\colon \Real^S \to  2^{\Pi \times \mathcal{P}}$ computes the saddle point policies and is defined as
\begin{equation} \label{eq:bellman-policy-robust} 
\mathfrak{B}\opt \bm{v}
  \; :=\; 
  \saddle( (\bm{\pi}, \bm{p}) \mapsto \mathfrak{T}^{\bm{\pi}, \bm{p}} \bm{v}),
\end{equation}
where the partial order on the value functions is defined as $\bm{u} \le \bm{v} \Leftrightarrow u_s \le v_s, \forall s\in \mathcal{S}$.

Bellman operators can be used to efficiently evaluate the value function. The value functions defined in~\eqref{eq:value-policy} and~\eqref{eq:value-optimal} are the \emph{unique} solution for each $\bm{\pi}\in \Pi$ and $\bm{p} \in \mathcal{P} $ to~\cite[corollary~10.1]{Kallenberg2022} 
\begin{equation*}
\bm{v}^{\bm{\pi}, \bm{p}} = \mathfrak{T} \bm{v}^{\bm{\pi}, \bm{p}},
  \qquad
  \bm{v\opt} = \mathfrak{T} \bm{v\opt}.
\end{equation*}
The Bellman operators $\mathfrak{T}^{\bm{\pi}, \bm{p}}$ and $\mathfrak{T}\opt$ are monotone and $\gamma$-contractive in the $L_{\infty}$ norm~\cite[theorem~10.5]{Kallenberg2022}. They can be computed by solving a linear program~\cite[appendix~C]{Ho2022} or by a specialized algorithm~\cite[section~6]{Ho2022}.

\section{Additional Material for Section~\ref{sec:filar-tolw-suboptimal}}
\label{sec:addit-counterexample}

\subsection{Polatzchek Avi-Itzhak Algorithm} \label{sec:pai}

Polatzchek Avi-Itzhak~(PAI) algorithm is a popular algorithm for solving MGs summarized in \cref{alg:PAI}. PAI works by starting with some arbitrary estimate of the value function for each state, and then alternating steps which evaluate the current policies of both agents, and then improve them based on their evaluations. PAI has been found to be very effective numerically, when it converges. However, PAI may never reach a sufficiently optimal solution, causing it to run forever~\cite{VanderWal1978, Condon1993}. Because it  difficult to know when the PAI converges it is unreliable in practice.

\begin{algorithm}
    \caption{Polatzchek Avi-Itzhak (PAI) Algorithm}
    \label{alg:PAI}
    \KwIn{Initial value $\bm{v}^0$, tolerance $\epsilon > 0$}
    \KwOut{$(\bm{\pi}, \bm{\sigma}) \in \saddle_{\epsilon}(\rhog)$}
    $k \gets 0$\;
    % ~~ $\bm{\pi}^k \gets \bm{\pi}_0$; ~~$\bm{\sigma}^k \gets \bm{\sigma}_0$; ~~ $\bm{v}^k \gets \bm{v}_0$\;
    \Repeat{$\frac{2 \gamma}{1-\gamma}\psi_\infty(\bm{v}^k) \le \epsilon$}
    {
        $k \gets k + 1$\;
        Select $(\bm{\pi}^k,\bm{\sigma}^k) \in \mathfrak{B}\opt \bm{v}^{k-1}$\;
        $\bm{v}^k \gets (\bm{I} - \gamma\bm{P}^{\bm{\pi}^k,\bm{\sigma}^k})^{-1}\bm{r}^{\bm{\pi}^k,\bm{\sigma}^k}$\;
    }
    \KwRet $(\bm{v}^k, \bm{\pi}^k,\bm{\sigma}^k)$\;
\end{algorithm}

\subsection{Proof of \cref{prop:rmdp-bound}}

\begin{proof}[Proof of \cref{prop:rmdp-bound}]
The result follows by algebraic manipulation as 
\begin{align*}
\min_{\bm{p}\in \mathcal{P}}  \rhor(\bm{\hat{\pi}}, \bm{p})
&\ge \rhor(\bm{\hat{\pi}}, \bm{\hat{p}}) - \epsilon  \\
&\ge \rhor(\bm{\pi\opt }, \bm{\hat{p}}) - 2\cdot \epsilon\\
&\ge \min_{\bm{p}\in \mathcal{P}} \rhor(\bm{\pi\opt }, \bm{p}) - 2\cdot \epsilon,
\end{align*}
where the first two inequalities follow from the saddle point definition in~\eqref{eq:saddle-point} and the last inequality from $\bm{\hat{p}} \in \mathcal{P}$.
\end{proof}

\subsection{Proof of \cref{thm:counter-example}}

\begin{proof}[Proof of \cref{thm:counter-example}]
Assume the MG in \cref{exm:local-minimum} and let $\bm{v}^0 = \bm{0} \in  \Real^3$. Basic algebraic manipulations show that
\[
  \mathfrak{T}\opt \bm{v}^0
  =
    \begin{bmatrix}
      -\sqrt{2} / 2\\ -1 / 2 \\ 1 / 2.
    \end{bmatrix}
\]
Since $v_2^0 = v_3^0 = 0$, each $\bm{\sigma}^1(s_1) \in \Delta^2$ is optimal in $\mathfrak{B}\opt(\bm{v}^0)$:
\begin{align}
     \bm{\sigma}^{1}(s_1) \in \argmin_{\bm{\sigma} \in \Delta^2} \;\; \bm{\sigma} \tr\begin{bmatrix} v_3 \\ v_2\end{bmatrix} = \Delta^2.
    \label{x_s1}
\end{align}
Consider the policy such that $\bm{\sigma}^1(s_1) = [1,0]$. Then, \cref{alg:FilarTolwinski} computes the step direction as:
\begin{align*}
  \bm{d}^1 &= (\bm{I}-\gamma\bm{P}^{\bm{\pi}^{1},\bm{\sigma}^{1}})^{-1}\bm{r}^{\bm{\pi}^{1},\bm{\sigma}^{1}} - \bm{v}^0 \\
           &= (\bm{I}-\gamma\bm{P}^{\bm{\pi}^{1},\bm{\sigma}^{1}})^{-1}(\mathfrak{T}\opt\bm{v}^0-\bm{v}^0).
\end{align*}
 Then, the algorithm seeks to compute $\bm{v}^1$ for some $\alpha = \beta^{i_k} \in (0,1]$
\begin{equation*}
  \bm{v}^1
  = 
  \bm{v}^0+\alpha \bm{d}^1
  =
    \begin{bmatrix}
    \alpha(\frac{3}{4}-\frac{\sqrt{2}}{2})\\
        -\alpha\frac{5}{4} \\
        \alpha\frac{5}{4}
    \end{bmatrix}
\end{equation*}
Since $v^1_2 < 0$ and $v^1_3 > 0$, the optimal policy satisfies
\begin{align*}
  \argmin_{\bm{\sigma} \in \Delta^2} \;\; \bm{\sigma}\tr\begin{bmatrix} v_3^0+\alpha d_3^1 \\ v_2^0+\alpha d_2^1
  \end{bmatrix}
    = \{[0,1]\},
\end{align*}
and, therefore,
\begin{align*}
    \psi_2(\bm{v}^0+\alpha \bm{d}^1)^2 - \psi_2(\bm{v}^0)^2 &= \left( \frac{3\sqrt{2}-4}{2} \right) \alpha + \left( \frac{13 - 6\sqrt{2}}{4} \right) \alpha^2,
\end{align*} 
which is always positive for $\alpha \in (0,1]$. Therefore, there is no descending step size along $\bm{d}^1$. \cref{alg:FilarTolwinski} then fails to find $i_k$ in Line~\ref{ln:ft-backtrack} and makes no progress even in the first iteration.
\end{proof}

\begin{remark} \label{exm:local-minimum2}  
To satisfy the starting constraint given from \cite[theorem~3.3]{Filar1991} that $\bm{v}^0 = r_{\max}\cdot \bm{1}$ where $r_{\max} = \max_{a\in \mathcal{A},b \in \mathcal{B},s'\in \mathcal{S}} |r(a,b,s')|$, consider a MG with the transition probabilities and rewards defined in \cref{tbl:counter-game2} and the discount factor $\gamma = 0.8$. 

% From \cref{exm:local-minimum2}
Then:
\begin{align*}
  \psi_2(\bm{v}^0+\alpha \bm{d}^1)^2 - \psi_2(\bm{v}^0)^2 = \frac{76}{25}\alpha + \frac{302}{25}\alpha^2
\end{align*}
which is always positive for $\alpha \in (0,1]$, thus using the initial value from \cite[theorem~3.3]{Filar1991} does not fix FT's convergence issues.
\end{remark}

\begin{figure}
  \centering
\begin{tabular}{|c|c|}
\hline
  \multicolumn{2}{|c|}{$s_1$} \\
  \hline
  $b_1$ & $b_2$ \\
  \hline
  $-1/2$ & $-1/2$ \\
$[0,0,1]$ & $[0,1,0]$ \\
\hline
\end{tabular}
\hspace{2cm}
\begin{tabular}{|c|}
\hline
  $s_2$ \\
  \hline
  $b_1$ \\
\hline
    $-1/2$ \\ 
 $[0,1,0]$ \\
\hline
\end{tabular}
\hspace{2cm}
\begin{tabular}{|c|}
\hline
  $s_3$ \\
  \hline
  $b_1$ \\
\hline
    $1/2$ \\ 
 $[0,0,1]$ \\
\hline
\end{tabular}
\hfill
\caption{Rewards and transition probabilities of the Markov game for states $s_1, s_2, s_3$ from \cref{exm:local-minimum2}.} \label{tbl:counter-game2}
\end{figure}

\section{Additional Material for Section~\ref{sec:rcpi:-resid-corr}}
\subsection{Proof of \cref{lem:GameBackupRuntime}}

\begin{proof}[Proof of \cref{lem:GameBackupRuntime}]
The matrix game for each state $s \in \mathcal{S}$ with a value function $\bm{v} \in \Real^{\mathcal{S}}$ is $\bm{G}^{\bm{v},s} \in \Real^{A \times B}$: 
\begin{align} \label{eq:MatrixGameGeneration}
  \bm{G}^{\bm{v},s}_{i,j}
  =
  \sum_{s'\in \mathcal{S}} p(s,i,j,s')\cdot \left( r(s,i,j,s') + \gamma \cdot \bm{v}(s') \right), \quad \forall i\in \mathcal{A}, j\in \mathcal{B}~.
\end{align}
Generating the matrix $\bm{G}^{\bm{v},s}$ can be done in $O(SAB)$ operations. To compute $\mathcal{B}\opt \bm{v}$, one needs to solve the saddle point problem
\begin{align*}
\max_{\bm{d} \in \Delta^\mathcal{A}} \min_{\bm{e} \in \Delta^\mathcal{B}} \bm{d}\tr \bm{G}^{\bm{v},s} \bm{e},
\end{align*}
which can be solved as the following linear program~\cite[theorem~10.1]{Kallenberg2022}:
\begin{equation*}
\min_{\bm{d'} \in \Real^A} \; \bm{1}\tr\bm{d'} \label{eq:MatrixGameLP}
\quad
\operatorname{subject\,to} \quad \bm{\tilde{G}}^{\bm{v},s}\bm{d'} \geq \bm{1} \nonumber ,
\bm{d'} \ge \bm{0}
\end{equation*}
where
\begin{align*}
    \tilde{\bm{G}}^{\bm{v},s}_{ij} =  \bm{G}^{\bm{v},s}_{ij} + 1 - \min_{a \in \mathcal{A}, b \in \mathcal{B}} \bm{G}^{\bm{v},s}_{ab}.
\end{align*}
This linear program contains $A + B$ constraints and $A$ variables, requiring $O\left((A + B)^{1.5}A^2\log(\delta^{-1}) \right)$ arithmetic operations to solve to a precision of $\delta$~\cite[section~10.1]{InteriorPointMethods}. Therefore, the total cost of computing $\mathfrak{B}^\delta \bm{v}$ and $\mathfrak{T}^\delta \bm{v}$ is the same as generating $\bm{G}^{\bm{v},s}$ and solving the linear program~\eqref{eq:MatrixGameLP} for every state, producing the bound in~\eqref{eq:MatrixGameRuntime}.
\end{proof}

\subsection{Proof of \cref{lem:IterationContraction}}

We first prove the following lemma that characterizes the convergence rate of value iteration with an approximate Bellman operator. 
\begin{lemma} \label{lem:ModifiedContraction}
Let $\{\bm{v}^k\}_{k=0}^\infty$ be a sequence where $\bm{v}^0 \in \Real^n$ and $\bm{v}^{k+1} = \mathfrak{T}^\delta\bm{v}^k$, then
\begin{align}
    \label{eq:ModifiedContraction}
  \psi_\infty^\delta(\bm{v}^{k+1})
  &\le \gamma\cdot \psi_\infty^\delta(\bm{v}^k) + 2(1 + \gamma)\delta, \\
    \label{eq:ContractiveInduction}
  &\le \gamma^{k+1} \cdot \psi_{\infty}^{\delta}(\bm{v}^0) + \frac{2 (1+\gamma )}{1-\gamma} \delta . 
\end{align}
\end{lemma}
\begin{proof}
  Let $\bm{d}^k := \mathfrak{T}^{\delta} \bm{v}^k - \mathfrak{T}\opt \bm{v}^k$ represent the approximation errors in $\mathfrak{T}^\delta\bm{v}^k$ then
    \begin{align*}
        \psi_\infty^\delta(\bm{v}^{k+1}) &\stackrel{\text{(a)}}{=} \|\mathfrak{T}\opt(\mathfrak{T}\opt\bm{v}^k + \bm{d}^k) + \bm{d}^{k+1} - (\mathfrak{T}\opt\bm{v}^k + \bm{d}^k)\|_\infty \\
        &\stackrel{\text{(b)}}{=} \|\mathfrak{T}\opt\mathfrak{T}\opt\bm{v}^k + \gamma\bm{P}\bm{d}^k + \bm{d}^{k+1} - \mathfrak{T}\opt\bm{v}^k - \bm{d}^k\|_\infty\\
        &\stackrel{\text{(c)}}{\leq} \gamma\|\mathfrak{T}\opt\bm{v}^k - \bm{v}^k\|_\infty + \|\bm{d}^{k+1}\|_\infty + (1 + \gamma)\|\bm{d}^{k}\|_\infty\\
        &\stackrel{\text{(d)}}{\leq} \gamma\|\mathfrak{T}^\delta\bm{v}^k - \bm{v}^k\|_\infty + \|\bm{d}^{k+1}\|_\infty + (1 + 2\gamma)\|\bm{d}^{k}\|_\infty\\
        &\stackrel{\text{(e)}}{\leq} \gamma\psi_\infty^\delta(\bm{v}^k) + 2(1+\gamma)\delta.
    \end{align*}
    The equality (a) follows from the definition of $\psi_\infty^\delta$,  (b) follows from the definition of $\mathfrak{T}\opt$ where $\bm{P}$ is a stochastic matrix. The inequality (c) comes from the triangle inequality, and (d) comes from the definition of $\mathfrak{T}^\delta$ combined with the triangle inequality. Finally, (e) follows from the definition of $\psi_\infty^\delta$ and the fact that $\delta \geq \|\bm{d}^k\|_\infty \;\;\forall k \in \Nats$.

     To show~\eqref{eq:ContractiveInduction} consider the tighter bound
     \begin{equation}
        \psi^\delta_\infty(\bm{v}^k) \leq \gamma^k\psi^\delta_\infty(\bm{v}^0) + \frac{2(1+\gamma)\delta(1-\gamma^k)}{1-\gamma}.
        \label{eq:ContractiveInductionExact}
     \end{equation}
    ~\eqref{eq:ContractiveInductionExact} is true for $k = 0$ since
    \[
        \gamma^0\psi^\delta_\infty(\bm{v}^0) + \frac{2(1+\gamma)\delta(1-\gamma^0)}{1-\gamma} = \psi^\delta_\infty(\bm{v}^0), 
    \]
    and if~\eqref{eq:ContractiveInductionExact} is true for $k$, then
    \begin{align*}
        \psi^\delta_\infty(\bm{v}^{k+1}) &\stackrel{\text{(f)}}{\leq} \gamma\left(\psi_\infty^\delta(\bm{v}^k)\right) + 2(1+\gamma)\delta\\
         &\stackrel{\text{(g)}}{\leq} \gamma\left(\gamma^k\psi^\delta_\infty(\bm{v}^0) + \frac{2(1+\gamma)\delta(1-\gamma^k)}{1-\gamma}\right) + 2(1+\gamma)\delta\\
         &\stackrel{\text{(i)}}{\leq} \gamma^{k+1}\psi^\delta_\infty(\bm{v}^0) + \left(\frac{2(1+\gamma)\delta(1-\gamma^{k+1})}{1-\gamma}\right).
    \end{align*}
    Inequality (f) follows from~\eqref{eq:ModifiedContraction}, (g) follows from~\eqref{eq:ContractiveInductionExact} being true for $k$, and (i) is a result of algebraic manipulation. Therefore~\eqref{eq:ContractiveInductionExact} is true for $k \in \Nats$ and since
    \[
        \gamma^{k}\psi^\delta_\infty(\bm{v}^0) + \left(\frac{2(1+\gamma)\delta(1-\gamma^{k})}{1-\gamma}\right) \leq \gamma^{k}\psi^\delta_\infty(\bm{v}^0) + \left(\frac{2(1+\gamma)\delta}{1-\gamma}\right) \;\;\forall k \in \Nats,
    \]
   ~\eqref{eq:ContractiveInduction} is also true for $k \in \Nats$.
\end{proof}

\begin{proof}[Proof of \cref{lem:IterationContraction}]
Recall that Line~\ref{ln:PolicySelection} requires $T_{\mathrm{G}}$ operations. Generating and solving the $S \times S$ system of linear equations on Line~\ref{ln:PolicyEvaluation} requires $O(S^2AB + S^3)$ operations, and the check on Line~\ref{ln:InitialResidualCheck} requires $T_{\mathrm{G}}$ operations since computing $\mathfrak{T}^\delta \bm{v}$ is the most expensive operation in $\psi^\delta_\infty(\bm{v})$.
   
The additional work done by \cref{alg:RCPI} on Line~\ref{ln:InitialResidualCheck} can be split into two cases.
    
\emph{Case 1:}
If $\gamma^{m-1}\psi_\infty^{\delta}(\bm{u}^{k,0}) + \frac{2(1 + \gamma)\delta}{1-\gamma} >  \;\psi_\infty^{\delta}(\bm{v}^{k-1})$, \cref{alg:RCPI} uses a Bellman backup to update the value function estimate by assigning $\mathfrak{T}^\delta\bm{v}^{k-1}$ to $\bm{v}^{k}$. Note that this step typically requires no significant additional work since $\mathfrak{T}^\delta \bm{v}$ can usually be computed from the results $\mathfrak{B}^\delta$ on Line~\ref{ln:PolicySelection}. Moreover, from \cref{lem:ModifiedContraction}, $\bm{v}^k$ satisfies~\eqref{eq:BellmanContraction}.

\emph{Case 2:}
If $\gamma^{m-1}\psi_\infty^{\delta}(\bm{u}^{k,0}) + \frac{2(1 + \gamma)\delta}{1-\gamma} \leq  \;\psi_\infty^{\delta}(\bm{v}^{k-1})$, we show that \cref{alg:RCPI} updates $\bm{v}^k$ to $\bm{u}^{k,l}$ for $l \le m$. In particular, we have that 
\begin{equation*} 
  \psi_{\infty}^{\delta}(\bm{u}^{k,m})
  \le \gamma^m \cdot \psi_{\infty}^{\delta}(\bm{u}^{k,0}) + \gamma \cdot \frac{2(1+\gamma)\delta}{1-\gamma}
  \le \gamma \cdot  \psi_\infty^{\delta}(\bm{v}^{k-1}).
\end{equation*}
The first inequality follows from \cref{lem:ModifiedContraction} and the last inequality follows from the  hypothesis of Case 2. Therefore, the loop on Line~\ref{ln:FixingCheck} terminates in at most $m$ iterations, runs in $O(m \cdot  T_{\mathrm{G}})$ operations, and $\bm{u}^{k,l}$ satisfies the termination condition $\psi_{\infty}^{\delta}(\bm{u}^{k,l}) \le  \gamma \cdot  \psi_{\infty}^{\delta}(\bm{v}^{k-1})$.
\end{proof}

\subsection{Proof of \cref{thm:RCPIRuntime}}

\begin{proof}[Proof of \cref{thm:RCPIRuntime}]
Recall from \cref{lem:ModifiedContraction,lem:IterationContraction}  it follows that
\begin{align}
    \psi^\delta_\infty(\bm{v}^k) &\leq \gamma^k\psi^\delta_\infty(\bm{v}^0) + \frac{2(1+\gamma)\delta}{1-\gamma} ~,
    \label{eq:ResidualDecreasingFunc}
\end{align}
and if $\bm{v}^0 = \bm{0}$ then
\begin{align}
    \psi^\delta_\infty(\bm{v}^0) \leq r_{\max} + \delta~.
    \label{eq:InitialValueErrorBound}
\end{align}
Let
\begin{equation} \label{eq:k-definition}
    k :=
    \left\lceil
    \frac{\log\left(\frac{1-\gamma}{2\gamma}\epsilon  -
        \frac{3+\gamma}{1-\gamma}\delta\right) -
        \log(r_{\max}+ \delta)}{\log(\gamma)} 
    \right\rceil,
\end{equation}
then
\begin{align*}
    \frac{2\gamma}{1-\gamma}(\psi_\infty^\delta(\bm{v}^k) + \delta)
    &\stackrel{\text{\eqref{eq:ResidualDecreasingFunc}}}{\le} \frac{2\gamma}{1-\gamma} \left( \gamma^k\psi_\infty^\delta(\bm{v}^0) + \frac{2(1+\gamma)\delta}{1-\gamma} + \delta \right)\\
    &\stackrel{\text{\eqref{eq:InitialValueErrorBound}}}{\le} \frac{2\gamma^{k+1}(r_{\max} + \delta)}{1-\gamma} +\frac{2\gamma(3+\gamma)\delta}{(1-\gamma)^2}\\
    &\stackrel{\text{\eqref{eq:k-definition}}}{\le} \epsilon
\end{align*}
The first equality follows from , the second inequality follows from , and the third inequality follows from algebraic manipulation. This guarantees that \cref{alg:RCPI} will terminate at or before iteration $k$, and when combined with the iteration runtime given in \cref{lem:IterationContraction}, produces the total runtime given in~\eqref{eq:RCPIRuntime}. The termination condition of \cref{alg:RCPI} guarantees the returned $(\bm{\pi},\bm{\sigma}) \in \saddle_\epsilon(\rho)$ as a direct result of \cref{prop:value-approximation-error}.
\end{proof}

\section{Additional Material for Section~\ref{sec:numerical-results}}

To create random Markov games for a given number of states, we randomly sample the number of actions available to each agent from $[1,2,3,5,10]$ for each state. The number of next states corresponding to each state-action-action combination is sampled as a fixed proportion $\eta$ of the total number of states. Following this method, we construct a series of games with varying state sizes and a common $\eta$ value of $\frac{1}{5}$.

Robust MDPs allowed the transition probabilities for each state to vary in absolute deviance by no more than 1. To generate instances of gambler's ruin, we use the maximum capital to control the number of states, and select the winning probability from a uniform distribution between 0 and 1. 

For random gridworld problems, the state size is controlled by specifying the side length, and the wind is randomly selected from a uniform distribution between 0 and 1. All the cells have a reward of 0 except one randomly picked cell with a reward of 1. 

To create random inventory problems for a given state count, a proportion of inventory states was randomly selected from a uniform distribution between 0 and 1. We assign the remaining states to the backlog, and the order capacity, storage cost, backlog cost, sale price, item cost, and delivery cost were all chosen randomly so that they do not conflict with one another or create bogus inventory problems (the sale price being less than the item cost, for example). The demand probabilities followed a Poisson distribution, where we sample the expected demand from a uniform distribution between 0 and the maximum demand.

\begin{figure*}
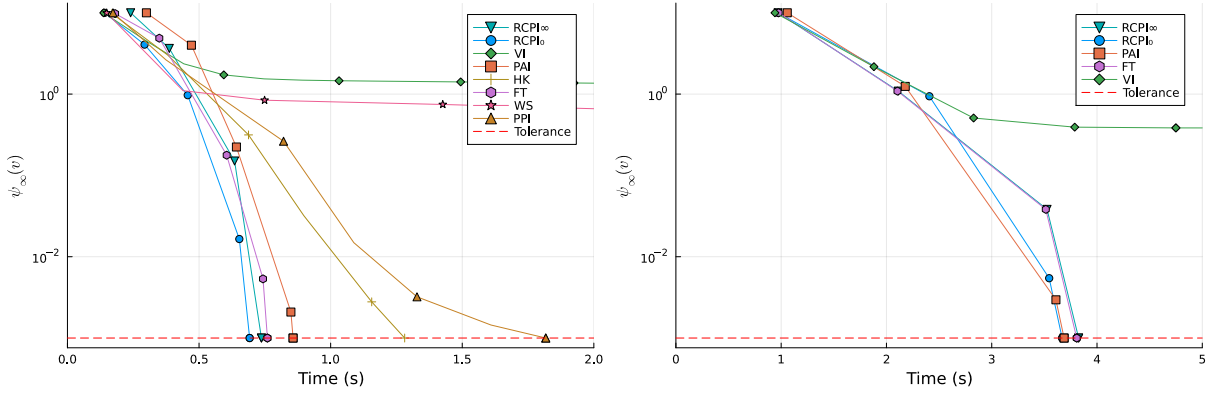

  \centering
    \includegraphics[width=.45\linewidth]{figures/mg\_all.pdf}
    \includegraphics[width=.45\linewidth]{figures/games\_large.pdf}
    \caption{The Bellman residual of each algorithm's value function plotted as a function of time for the smaller MGs \emph{(left)} with 20 to 100 states, and the larger MGs \emph{(right)} with 200 to 1000 states.}
    \label{fig:MarkovGames}
\end{figure*}

\begin{figure}
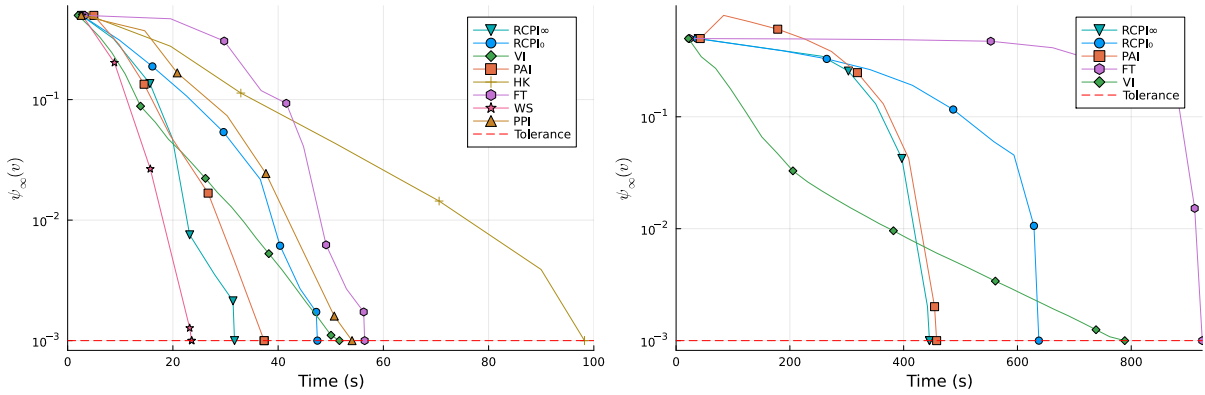

  \centering
        \includegraphics[width=.45\linewidth]{figures/ruin\_all.pdf}
        \includegraphics[width=.45\linewidth]{figures/ruin\_large.pdf}
    \caption{The Bellman residual of each algorithm's value function plotted as a function of time for the smaller gambler's ruin problems (\emph{left}) with 10 to 50 states, and the larger problems (\emph{right}) with 200 to 1000 states.}
    \label{fig:GamblersRuin}
\end{figure}

\begin{figure}
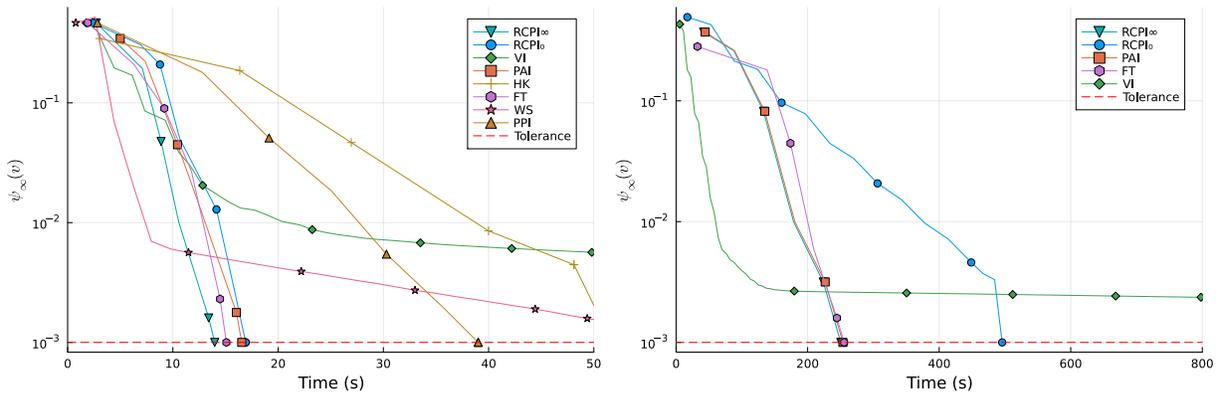

  \centering
        \includegraphics[width=.45\linewidth]{figures/grid\_all.pdf}
        \includegraphics[width=.45\linewidth]{figures/grid\_large.pdf}
    \caption{The Bellman residual of each algorithm's value function plotted as a function of time for the smaller gridworld problems (\emph{left}) with 4 to 100 states, and the larger problems (\emph{right}) with 16 to 400 states.}
    \label{fig:GridWorld}
\end{figure}

\begin{figure}
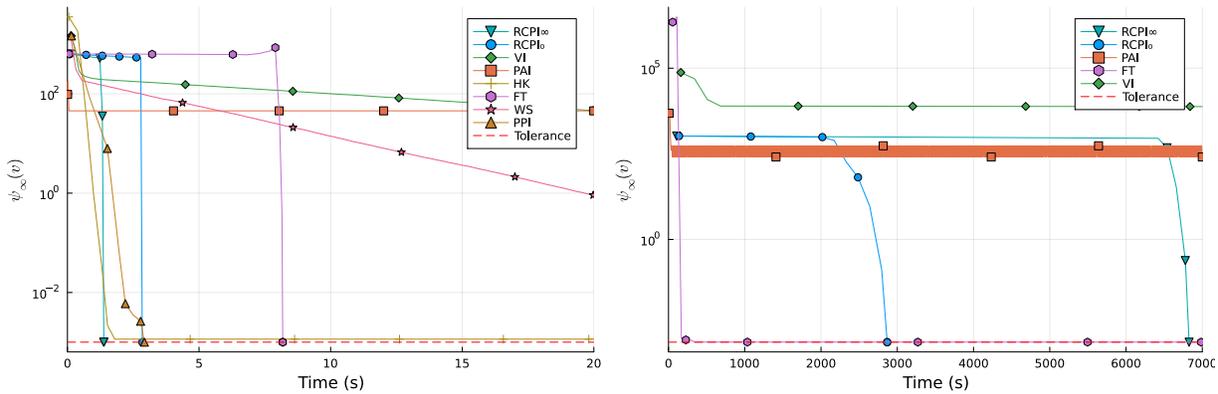

  \centering
        \includegraphics[width=.45\linewidth]{figures/inv\_all.pdf}
        \includegraphics[width=.45\linewidth]{figures/inv\_large.pdf}
    \caption{The Bellman residual of each algorithm's value function plotted as a function of time for the smaller inventory problems (\emph{left}) with 4 to 20 states, and the larger problems (\emph{right}) with 40 to 200 states.}
    \label{fig:Inventory}
\end{figure}

% \fi
\end{document}